\documentclass{svjour3}

\usepackage{amsfonts}
\usepackage{amsmath}
\usepackage{amssymb}
\usepackage{mathtools}
\usepackage{graphicx}
\usepackage{bm}
\usepackage[authoryear]{natbib}

\usepackage{graphicx}
\usepackage{mathptmx}      

\usepackage{lineno}

\smartqed

\journalname{Journal of Mathematical Biology}

\title{The stationary distribution of a sample from the Wright-Fisher diffusion model with general small mutation rates}

\titlerunning{Wright-Fisher diffusion with small mutation rates}

\author{Conrad J.\ Burden \and Robert C. Griffiths}
\institute{Conrad J.\ Burden, Mathematical Sciences Institute, Australian National University, Canberra, Australia; Research School of Biology, Australian National University, Canberra, Australia
\at \email{conrad.burden@anu.edu.au}
\and
Robert C.\ Griffiths,Department of Statistics, University of Oxford, UK;  Corresponding author
\at
\email{griff@stats.ox.ac.uk}
}

\date{Received: date / Accepted: date}

\begin{document}
\maketitle
\begin{abstract}
The stationary distribution of a sample taken from a Wright-Fisher diffusion with general small mutation rates is found using a coalescent approach. The approximation is equivalent to having at most one mutation in the coalescent tree to the first order in the rates. The sample probabilities characterize an approximation for the stationary distribution from the Wright-Fisher diffusion. The approach is different from \citet{BT2016,BT2017} who use a probability flux argument to obtain the same results from a forward diffusion generator equation.
The solution has interest because the solution is not known when rates are not small. 

An analogous solution is found for the configuration of alleles in a general exchangeable binary coalescent tree. In particular an explicit solution is found for a pure birth process tree when individuals reproduce at rate $\lambda$. 
\end{abstract}

\section{Introduction}
\citet{BT2016,BT2017} find an approximation for the stationary distribution in a $d$-allele neutral Wright-Fisher diffusion with low mutation rates. This is interesting to find, particularly since the solution of the stationary distribution for such a diffusion is unknown. For low mutation rates they find that approximately either the population is fixed for one allele type, or lies on a line density with just two types.  

This scenario has biological relevance in the context of the infinite sites model \citep{K1969}.  \cite{Z2010}, for instance, has estimated mutation rates between codons in the Drosophila genome.  Zeng's calculation is based on numerically determining the stationary distribution of a multi-allelic model which assumes segregating sites to have at most two variants if mutation rates are small.  Numerical simulations illustrated in Figs.~5 and 6 of \citet{BT2016} demonstrate that this assumption 
is accurate if the mutation rate $\theta$, as defined below Eq.~(\ref{gen:0}) in the current paper, is $\lesssim 0.01$.  Estimates of mutation rates in small introns 
and fourfold degenerate sites in Drosophila are close to this limit \citep{VC2012}, while mutation rates in vertebrates are generally considered to be somewhat smaller.  
A more extensive discussion of potential biological applications of such models can be found in the Discussion and Conclusions of \citet{BT2017}. 

Burden and Tang's method of solution is analytic and relies on parameterising the general non-reversible rate matrix as the sum of a reversible part and a set of $(d-1)(d-2)/2$ independent terms corresponding to fluxes of probability along closed paths around faces of the simplex.
The model has a connection to boundary processes which take only these two types of states with mutation moving a fixed population to a configuration with two allele types and then no mutation taking place until after fixation occurs again.
\citet{SH2017}, who cite \citet{VB2015,DSK2015}, study a Moran type model
with this boundary mutation property. The stationary distribution is shown to be similar to the approximate solution to a full Moran model with low mutation rates.

The Kingman coalescent process is dual to the Wright-Fisher diffusion in describing the ancestral history of a sample of $n$ individuals in a population back in time. 
In this paper a coalescent approach with mutations in the tree is used to find an approximate sampling formula for small mutation rates. The sampling distributions in the coalescent uniquely determine the stationary distribution in the diffusion process because of duality. Although the focus here is on sampling distributions in the coalescent the approximations give unique expressions for approximations in the stationary distribution of the diffusion. We re-derive a formula of \citet{BT2016,BT2017} using the coalescent. If the mutation rates are of order $\theta$ then it turns out that finding approximate formulae for small rates is equivalent to considering at most one mutation in a coalescent tree.
The method of proof is considerably simpler than the original proofs in \citet{BT2016,BT2017}. A second proof is given using the backward generator of the diffusion process.
The idea of deriving approximate sampling formulae in population genetics models with small mutation rates (or other small parameters such as the inverse of the recombination rate) by considering the number of mutation events in a marked coalescent tree is natural and has been used before, for example in \citet{JS2010,JS2011,BKS2012}.

\citet{BG2018} find an approximation to the stationary density in a two island, two allele model when mutation and migration rates are small by using a flux argument in the Wright-Fisher diffusion process as well as a coalescent argument.

\citet{GT1988} study general binary coalescent trees which have an exchangeable coalescence structure. The method of proof for small mutation rates in the Kingman coalescent tree is extended to general binary coalescent trees in a new Theorem \ref{thm:2}. A general formula for the sample configuration in the ancestral tree of a pure birth process follows from Theorem \ref{thm:2}. 
The formula has an interesting specific form when  individuals in the birth process reproduce independently at rate $\lambda$.

%
%
\section{The Wright-Fisher diffusion}

A model of the relative frequency of $d$ alleles is a Wright-Fisher diffusion process $\{\bm{X}(t)\}_{t\geq 0}$. We consider a neutral model which with total mutation rate $\theta/2$ and a irreducible transition matrix for mutation changes between types $P$, which has a stationary distribution $\bm{\pi}$. The backward generator of the diffusion process is 
\begin{equation}
\mathcal{L}=\frac{1}{2}\sum_{i,j=1}^dx_i(\delta_{ij}-x_j)\frac{\partial^2}{\partial x_i\partial x_j}+\sum_{i,j=1}^d\gamma_{ji}x_j\frac{\partial}{\partial x_i}.
\label{gen:0}
\end{equation}
The mutation rates are parameterized as $\gamma_{ij}=\frac{\theta}{2}(P_{ij}-\delta_{ij})$. There is no loss of generality in taking the total mutation rate $\theta/2$, because it is possible to take the diagonal of $P$ to have non-zero entries, effectively allowing different overall rates away from types.

Let $\bm{n}=(n_1,\ldots,n_d)$ be the number of genes of types $1,\ldots, d$ in a sample of $n$ genes taken from the population.
By a dual process argument the sampling distribution in the stationary distribution of the diffusion of $n$ individuals 
\begin{equation}
p(\bm{n};\theta) = {n\choose \bm{n}}\mathbb{E}\big [X_1^{n_1}\cdots X_d^{n_d}\big ].
\label{dsample:0}	
\end{equation}
is the same as the distribution of alleles in the leaves of a coalescent tree of $n$ individuals. A brief aspect of this duality is the following.  In the stationary distribution of the diffusion process
\begin{equation}
\mathbb{E}\Big [\mathcal{L}{n\choose \bm{n}}X_1^{n_1}\cdots X_d^{n_d}\Big ] = 0.
\label{stat:0}
\end{equation}
In general if a Markov process has a generator ${\cal L}$ then for $f$ in the domain of ${\cal L}$ 
\[
\mathbb{E}\big [{\cal L}f(\bm{X})\big ] = 0,
\]
with expectation in the stationary distribution of the process if it exists. See, for example, \citet{E2011} p46.
A recursive equation for the sampling probabilities is implied from (\ref{stat:0}).
Now
\begin{eqnarray}
&&\mathcal{L}x_1^{n_1}\cdots x_i^{n_i}\cdots x_d^{n_d} =
\nonumber \\
&& \frac{1}{2}\sum_{i=1}^dn_i(n_i-1)x_1^{n_1}\cdots x_i^{n_i-1}\cdots x_d^{n_d}
-\frac{1}{2}\sum_{i,j=1}^dn_i(n_j-\delta_{ij})x_1^{n_1}\cdots x_d^{n_d}
\nonumber\\
\nonumber\\
&&+\frac{\theta}{2}\sum_{i,j=1}^dP_{ji}n_i x_1^{n_1}\cdots x_i^{n_i-1+\delta_{ij}}\cdots x_j^{n_j+1-\delta_{ij}}\cdots  x_d^{n_d}  
-\frac{\theta}{2}\sum_{i=1}^dn_ix_1^{n_1}\cdots x_d^{n_d}.
\nonumber \\
\label{stat:1}
\end{eqnarray}
Simplifying (\ref{stat:1}) and using (\ref{stat:0}) 
\begin{eqnarray}
p(\bm{n};\theta) &=& \frac{n-1}{n-1 +\theta}\sum_{i=1}^n\frac{n_i-1}{n-1}p(\bm{n}-\bm{e}_i;\theta)
\nonumber \\
&&
+ \frac{\theta}{n-1 + \theta}\sum_{i,j=1}^d\frac{n_j+1-\delta_{ij}}{n}P_{ji}p(\bm{n}-\bm{e}_i + \bm{e}_j;\theta).
\label{stat:2}
\end{eqnarray}
The boundary conditions are that $p(\bm{e}_i;\theta) = \pi_i$.
The recursion (\ref{stat:2}) is well known, see, for example \citet{DG2004} Eq.~(3).

An alternative coalescent argument to derive (\ref{stat:2}) is that in a coalescent tree of $n$ where mutations occur on the edges, conditional on the edge lengths, according to a Poisson process of rate $\theta/2$ the probability that the first event back in time was a coalescence is
\[
\frac{{n\choose 2}}{{n\choose 2} + n\theta/2} = \frac{n-1}{n-1+\theta}.
\]
The probability that the first event back in time was a mutation is similarly
\[
\frac{\theta}{n-1+\theta}.
\]
If the event was a coalescence, then the probability of obtaining a configuration of $\bm{n}$ from $\bm{n}-\bm{e}_i$ is $(n_i-1)/(n-1)$. If the event was a mutation, the probability of obtaining a configuration $\bm{n}$ from $\bm{n} - \bm{e}_i + \bm{e}_j$ is
$P_{ji}(n_j+1)/n$ if $i\ne j$ or $P_{ii}n_i/n$ if $i=j$.

Therefore calculating the probability of a sample configuration from (\ref{dsample:0}) is the same as calculating the probability of a configuration of $\bm{n}$ in a coalescent tree.

The emphasis in this paper is finding an expression for $p(\bm{n};\theta)$ when the mutation rate $\theta$ is small. That is, to find a formula
\[
p(\bm{n};\theta) = q_0(\bm{n}) + q_1(\bm{n})\theta+ \mathcal{O}(\theta^2)
\]
as $\theta \to 0$ using a coalescent approach. Then $\bm{n}$ is a configuration of a single allele type, the ancestor in the coalescent tree, if there is no mutation; or a single mutation from the ancestor type to itself. $\bm{n}$ is a configuration of two different allele types if there is a single mutation from the ancestor type to a different type  in the coalescent tree.

A preliminary lemma that is needed is the following, from \citet{GT1988}, Eq. (1.9). The Lemma applies in general exchangeable binary trees where coalescent times $T_2,\ldots, T_n$ have a general distribution.
\begin{lemma}
A particular edge when there are $k$ edges in a general exchangeable binary tree subtends $c$ leaves in the $n$ leaves of a coalescent tree with probability
\begin{equation}
p_{nk}(c) = \frac{{n-c-1\choose k-2}}{{n-1\choose k-1}},\>k \leq n-c+1.
\end{equation}
\end{lemma}
This is a Polya urn result identifying $k-1$ edges as balls of one colour and the particular edge as a ball of another colour. Branching in the tree is identified with drawing a ball and replacing it together with another of the same colour. A classical de Finetti representation is
\[
p_{nk}(c) = \int_0^1{n-k\choose c-1}z^{c-1}(1-z)^{n-c-k+1}\times (k-1)(1-z)^{k-2}dz
\]
for which see, for example, \citet{GT2003}, Eq.~(2.4).


The next theorem is an analogous sample version, with a coalescent proof, of the approximation to the stationary sampling density, Eq.~(35) of \citet{BT2017}.  
Certain special cases of this density, corresponding to situations where the stationary distribution of the Wright-Fisher or Moran diffusion is known exactly, 
had been previously published.  The corresponding 2-allele case is quoted in Eq.~(29) of \citet{V2014}, and the case of multi-allelic 
parent-independent rate matrix is given in Eq.~(10) of \citet{RW2010}.  Both of these special cases correspond to reversible rate matrices, for which 
$\pi_a P_{ab} = \pi_b P_{ba}$, leading to a symmetry in Eq.~(\ref{thm:b}).  Importantly, this symmetry is not present for a general rate matrix.

\begin{theorem}\label{thm:1}
The probability of sample configuration in the stationary distribution of the Wright-Fisher diffusion with generator (\ref{gen:0}) is the same as the distribution of a sample configuration in a coalescent tree. As $\theta \to 0$, for $a\ne b \in \{1,\ldots,d\}$, and $n_a+n_b=n$,
\begin{eqnarray}
p(n\bm{e}_a;\theta) &=& \pi_a\Big (1 - \theta(1-P_{aa})\sum_{l=1}^{n-1}\frac{1}{l}\Big ) + \mathcal{O}(\theta^2)
\label{thm:a} \\
&=&\pi_a\Big (1 + 2\gamma_{aa}\sum_{l=1}^{n-1}\frac{1}{l}\Big ) + \mathcal{O}(\theta^2)
\nonumber \\
p(n_a\bm{e}_a+n_b\bm{e}_b;\theta) &=& \theta\Big (\pi_aP_{ab}\frac{1}{n_b} + \pi_bP_{ba}\frac{1}{n_a}\Big ) + \mathcal{O}(\theta^2)
\label{thm:b} \\
p(\bm{n};\theta) &=& \mathcal{O}(\theta^2)\text{~if~}\bm{n}\text{~has~}>2\text{~non-zero~entries}
\label{thm:c}
\end{eqnarray}
\end{theorem}
\begin{proof}
Let $T_2, \ldots, T_n$ be the times while $2,3,\ldots,n$ ancestors in a coalescent tree of $n$ individuals. These are independent exponential random variables with rates ${2\choose 2},\ldots, {n\choose 2}$.
The probability that a sample of $n$ is monomorphic of type $a$, and there are no mutations in the tree is
\begin{eqnarray}
\pi_a\mathbb{E}\Big [\exp \Big \{-\frac{\theta}{2}\sum_{l=2}^nlT_l \Big \}\Big ]
&=& \pi_a\Big ( 1 - \frac{1}{2}\theta\sum_{l=2}^{n}l\mathbb{E}\big [T_l\big ]\Big ) + \mathcal{O}(\theta^2)
\nonumber \\
&=& \pi_a\Big ( 1 - \theta\sum_{l=1}^{n-1}\frac{1}{l}\Big ) + \mathcal{O}(\theta^2).
\label{thm:aa}
\end{eqnarray}
Mutations occur according to a Poisson process along the edges of the tree, conditional on the coalescence times $T_n,\ldots, T_2$.  The conditional probability of a single mutation occurring while $k$ edges in the tree is therefore
\[
\frac{\theta}{2}kT_ke^{-\frac{\theta}{2}\sum_{l=2}^nlT_l}.
\]
The probability that a sample of $n$ is monomorphic of type $a$ and one mutation from a type $a$ allele to the same type in the tree is
\begin{eqnarray}
\pi_aP_{aa}\frac{\theta}{2}\sum_{l=2}^n
 l\mathbb{E}\big [T_le^{-\frac{\theta}{2}\sum_{l=2}^nlT_l}\big ]
  &=& 
\pi_aP_{aa}\frac{\theta}{2}\sum_{l=2}^nl\mathbb{E}\big [T_l\big ] + \mathcal{O}(\theta^2)
\nonumber \\
&=& \theta\pi_aP_{aa}\sum_{l=1}^{n-1}\frac{1}{l} + \mathcal{O}(\theta^2)
\label{thm:ab}
\end{eqnarray}
Adding (\ref{thm:aa}) and (\ref{thm:ab}) gives the probability of a monomorphic configuration (\ref{thm:a}).

The probability of a configuration $n_a\bm{e}_a+n_b\bm{e}_b$, $a,b\in [d]$ is now calculated. This configuration can only be obtained to order $\theta$ if there is one mutation in the tree. 
The probability that the ancestor is of type $a$ and 
 a single mutation occurs while $k$ edges giving rise to a type $b$ individual conditional on the coalescent times is
\[
\pi_aP_{ab}\frac{\theta}{2}kT_k\exp \Big \{-\frac{\theta}{2}\sum_{l=2}^nlT_l\Big \}.
\]
The total probability of a configuration $n_a\bm{e}_a + n_b\bm{e}_b$ with a type $a$ ancestor is then
\begin{eqnarray}
&&\frac{\theta}{2}\pi_aP_{ab}\sum_{k=2}^{n_a+1}\mathbb{E}\big [T_k]kp_{nk}(n_b) + \mathcal{O}(\theta^2)
\nonumber \\ 
&&~=\theta\pi_aP_{ab}\frac{(n_a-1)!(n_b-1)!}{(n-1)!}
\sum_{k=2}^{n_a+1}{n-k \choose n_b-1} + \mathcal{O}(\theta^2).
\label{thm:ba}	
\end{eqnarray}
The sum in (\ref{thm:ba}) is the coefficient of $z^{n_b-1}$ in 
\[
\sum_{k=2}^{n_a+1}(1+z)^{n-k}=\sum_{l=n_b-1}^{n-2}(1+z)^l
= (1+z)^{n_b-1}\frac{(1+z)^{n_a}-1}{z},
\]
that is ${n-1\choose n_b}$. The total probability of the event is then
\[
\theta\pi_aP_{ab}\frac{(n_a-1)!(n_b-1)!}{(n-1)!}\cdot \frac{(n-1)!}{n_b!(n_a-1)!} + \mathcal{O}(\theta^2) = \theta\pi_aP_{ab}\frac{1}{n_b} + \mathcal{O}(\theta^2).
\]
Similarly if the ancestor is of type $b$ the probability is 
\[
\theta\pi_bP_{ba}\frac{1}{n_a} + \mathcal{O}(\theta^2).
\]
 The full probability of the configuration taking into account whether the ancestor is type $a$ or $b$ is then (\ref{thm:b}).
\qed
\end{proof}
A second induction proof can be made using the backward generator (\ref{gen:0}).

\begin{proof}

We want to prove that in the stationary distribution of the diffusion process with generator (\ref{gen:0}) that when $\theta \to 0$ the sample probabilities satisfy
\begin{eqnarray}
\mathbb{E}\big [X_a^{n}] &=& 
\pi_a\Big (1 - \theta(1-P_{aa})\sum_{l=1}^{n-1}\frac{1}{l}\Big ) + \mathcal{O}(\theta^2)
\label{g:a}\\
{n \choose n_a}\mathbb{E}\big [X_a^{n_a}X_b^{n_b}\big ] &=& 
\theta\pi_aP_{ab}\frac{1}{n_b}+\theta\pi_bP_{ba}\frac{1}{n_a} + \mathcal{O}(\theta^2)
\label{g:ab}
\end{eqnarray}
 For $n>1$
\begin{eqnarray}
0 &=& \mathbb{E}\big [{\cal L}X_a^{n}\big ]	
\nonumber \\
&=& \frac{1}{2}n(n-1)\mathbb{E}\big [X_a^{n-1}- X_a^{n}\big ]+n\frac{\theta}{2}\sum_{j=1}^d(P_{ja}-\delta_{ja})\mathbb{E}\big [X_a^{n-1}X_j\big ].
\label{temp:20}
\end{eqnarray}
$\mathbb{E}\big [X_a\big ] = \pi_a$ and from (\ref{temp:20})
$\mathbb{E}\big [X_a^n\big ] - \mathbb{E}\big [X_a^{n-1}\big ] = \mathcal{O}(\theta)$, so by recurrence 
 $\mathbb{E}\big [X_a^n\big ] = \pi_a + \mathcal{O}(\theta)$.
An induction proof now shows that (\ref{g:a}) is true. This is a refinement making calculations to $\mathcal{O}(\theta^2)$. The formula is correct for $n=n_a=1$ since $\mathbb{E}\big [X_a\big ] = \pi_a$. Assume that (\ref{g:a}) is true for $n$ replaced by $n-1$.
Simplifying (\ref{temp:20}) gives
\begin{equation}
\big (n-1+\theta(1-P_{aa})\big )\mathbb{E}\big [X_a^n\big ]
= (n-1)\mathbb{E}\big [X_a^{n-1}\big ]
+ \theta \sum_{j\ne a}P_{ja}\mathbb{E}\big [X_a^{n-1}X_j\big ].
\label{temp:10}
\end{equation}
The last term in (\ref{temp:10}) is $\mathcal{O}(\theta^2)$ because 
\begin{eqnarray*}
\theta \sum_{j\ne a}P_{ja}\mathbb{E}\big [X_a^{n-1}X_j\big ]
&\leq& \theta\sum_{j\ne a}\mathbb{E}\big [X_aX_j\big ]
\nonumber \\
&=&\theta\big (\mathbb{E}\big [X_a] - \mathbb{E}\big [X_a^2\big ]\big )
\nonumber \\
&=& \theta(\pi_a-\pi_a + \mathcal{O}(\theta)) = \mathcal{O}(\theta^2)
\end{eqnarray*}
This estimate is valid even if $n=2$ since it is known at this stage of the proof that 
$\mathbb{E}\big [X_a^2\big ] = \pi_a+\mathcal{O}(\theta)$.
Now from (\ref{temp:10}) and the induction hypothesis
\begin{eqnarray*}
\mathbb{E}\big[X_a^n\big ] &=& \frac{n-1}{n-1+ \theta(1-P_{aa})}\mathbb{E}\big[X_a^{n-1}\big ] + \mathcal{O}(\theta^2)
\nonumber \\
&=& \frac{n-1}{n-1+\theta(1-P_{aa})}\pi_a\Big (1 - \theta(1-P_{aa})\sum_{l=1}^{n-2}\frac{1}{l}\Big ) + \mathcal{O}(\theta^2)
\nonumber \\
&=&
\pi_a\Big (1 - \theta(1-P_{aa})\sum_{l=1}^{n-1}\frac{1}{l}\Big ) + \mathcal{O}(\theta^2)
\end{eqnarray*}
and the induction is completed.

For (\ref{g:ab}) consider for $n_a\geq 1,n_b \geq 1, n \geq 2$
\begin{eqnarray}
0&=& \mathbb{E}\big [{\cal L}X_a^{n_a}X_b^{n_b}\big ]
\nonumber \\
&=& \frac{1}{2}n_a(n_a-1)\mathbb{E}\big [X_a^{n_a-1}X_b^{n_b}\big ]
+ \frac{1}{2}n_b(n_b-1)\mathbb{E}\big [X_a^{n_a}X_b^{n_b-1}\big ]
\nonumber \\
&&~~-\frac{1}{2}n(n-1)\mathbb{E}\big [X_a^{n_a}X_b^{n_b}\big ]
\nonumber \\
&&~~+ n_a\sum_{j = 1}^d \gamma_{ja} \mathbb{E}\big [X_j X_a^{n_a - 1}X_b^{n_b}\big ]
	+ n_b\sum_{j = 1}^d \gamma_{jb} \mathbb{E}\big [X_j X_a^{n_a}X_b^{n_b - 1}\big ].  \label{LofXaXb}
\end{eqnarray}

The proof proceeds by induction.  Consider first the case $n_a = n_b = 1$.  In this case the first two terms in Eq.~(\ref{LofXaXb}) do not contribute, and 
\begin{eqnarray*}
\mathbb{E}\big [X_aX_b\big ] & = & \frac{\theta}{2}\Bigg (\sum_{j=1}^d(P_{ja}-\delta_{ja})\mathbb{E}\big [X_jX_b\big ]
+ \sum_{j=1}^d(P_{jb}-\delta_{jb})\mathbb{E}\big [X_jX_a\big ]\Bigg )
\nonumber \\
&=&\frac{\theta}{2}\Big ( P_{ba}\mathbb{E}\big [X_b^2\big ]
+P_{ab}\mathbb{E}\big [X_a^2\big ]\Bigg ) + \mathcal{O}(\theta^2)
\nonumber \\
&=& \frac{\theta}{2}\Big ( P_{ba}\pi_b+P_{ab}\pi_a\Big ) + \mathcal{O}(\theta^2),
\label{start:0}
\end{eqnarray*}
which establishes Eq.~(\ref{g:ab}) for $n_a = n_b = 1$.  

Now suppose $n_a = 1$, and carry out an induction on $n_b > 1$.  From Eq.~(\ref{LofXaXb}), 
\begin{eqnarray*}
\lefteqn{\frac{1}{2}n(n-1)\mathbb{E}\big [X_aX_b^{n_b}\big ] } \\
      & = &\frac{1}{2}n_b(n_b-1)\mathbb{E}\big [X_a X_b^{n_b-1}\big ] \\
		& & + \frac{\theta}{2} \left( \sum_{j = 1}^d (P_{ja}-\delta_{ja}) \mathbb{E}\big [X_j X_b^{n_b}\big ] + 
					n_b \sum_{j = 1}^d (P_{jb}-\delta_{jb}) \mathbb{E}\big [X_j X_a X_b^{n_b - 1}\big ] \right) \\
      & = &\frac{1}{2}n_b(n_b-1)\mathbb{E}\big [X_a X_b^{n_b-1}\big ] + \frac{\theta}{2} P_{ba}\pi_b + \mathcal{O}(\theta^2).  
\end{eqnarray*}
Assume that (\ref{g:ab}) holds for $n_a = 1$, $n_b$ replaced by $n_b - 1$, and $n$ replaced by $n - 1 = n_b$, that is, 
$$
n_b \mathbb{E}\big [X_a X_b^{n_b-1}\big ] = \theta \pi_a P_{ab} \frac{1}{n_b - 1} + \theta \pi_b P_{ba} + \mathcal{O}(\theta^2).  
$$
Then 
\begin{eqnarray*}
\lefteqn{{n \choose 1} \mathbb{E}\big [X_aX_b^{n_b}\big ]} \\
 & = & \frac{2}{n_b} \left[ \frac{1}{2}(n_b-1) \left(\theta \pi_a P_{ab} \frac{1}{n_b - 1} + \theta \pi_b P_{ba} \right) 
		+ \frac{\theta}{2} P_{ba}\pi_b \right] + \mathcal{O}(\theta^2) \\
 & = & \theta \pi_a P_{ab} \frac{1}{n_b} + \theta \pi_b P_{ba} + \mathcal{O}(\theta^2).  
\end{eqnarray*} 
This establishes (\ref{g:ab}) for $n_a = 1$ and $n_b > 1$.  The cases $n_a > 1$ and $n_b = 1$ follow by symmetry.

Finally, for both $n_a, n_b > 1$, Eq.~(\ref{LofXaXb}) rearranges to give  
\begin{eqnarray}
{n\choose n_a}\mathbb{E}\big [X_a^{n_a}X_b^{n_b}\big ]
&=& \frac{n_a-1}{	n-1}{n-1\choose n_a-1}\mathbb{E}\big [X_a^{n_a-1}X_b^{n_b}\big ]
\nonumber \\
&&~~ +\frac{n_b-1}{	n-1}{n-1\choose n_b-1}\mathbb{E}\big [X_a^{n_a}X_b^{n_b-1}\big ] + \mathcal{O}(\theta^2).  
\label{recurse:0}
\end{eqnarray}
Assume (\ref{g:ab}) is true up to $n-1$. Then
\begin{eqnarray}
{n\choose n_a}\mathbb{E}\big [X_a^{n_a}X_b^{n_b}\big ] &=&
\theta\cdot \frac{n_a-1}{n-1}	\Bigg (\pi_bP_{ab}\frac{1}{n_b}+\pi_bP_{ba}\frac{1}{n_a-1}\Bigg )
\nonumber \\
&&~+
\theta\cdot \frac{n_b-1}{n-1}	\Bigg (\pi_bP_{ab}\frac{1}{n_b-1}+\pi_bP_{ba}\frac{1}{n_a}\Bigg )
\nonumber \\
&=& \theta\Bigg (\pi_bP_{ab}\frac{1}{n_b}+\pi_bP_{ba}\frac{1}{n_a}\Bigg ).
\end{eqnarray}
The induction proof is now complete.
\qed

\end{proof}
The stationary density in the population $f(\bm{x};\theta)$ is singular as $\theta \to 0$. There is a probability $\pi_a + \mathcal{O}(\theta)$ that the population is fixed for type $a$. An assumption needs to be made that the diffusion is an approximation to a discrete Wright-Fisher model of population size $N$ and $\theta\log (N)<<1$ in (\ref{thm:a}). 
There are at most two types $a$ and $b$ to $\mathcal{O}(\theta)$ and the density of $X_a$ and $X_b=1-X_a$ is proportional to
\begin{equation}
\theta\big (\pi_aP_{ab}\frac{1}{1-x_a} + \pi_bP_{ba} \frac{1}{x_a}\big ),
\label{density:20}
\end{equation}
with $1/N < x_a < 1-1/N$, found by \citet{BT2016}. Alternatively (\ref{density:20}) is an approximation that holds for 
$\theta|\log x_a + \log (1-x_a)| << 1$.
An important step in their flux argument is that the $d\times d$ rate matrix $Q=\frac{1}{2}\theta(P-I)$ can be decomposed as $Q=Q^{\text{GTR}}+Q^{\text{flux}}$ where $Q^{\text{GTR}}= (C_{ij})$ is a general time-reversible rate matrix and $Q^{\text{flux}}= (\Phi_{ij})$ is a matrix with elements satisfying $\Phi_{ij}=-\Phi_{ji}$.
Their approximate density has the following form (Eq. (42), (2016) paper and Eq. (21), (2017) paper)
\[
(C_{ab}-\Phi_{ab})\frac{1}{x}+(C_{ab}+\Phi_{ab})\frac{1}{1-x},
\]
which turns out to be proportional to (\ref{density:20}). The probability that the population is monomorphic is Eq. (45) in \citet{BT2016}.
If $P$ is reversible then the density (\ref{density:20}) is proportional to
\[
\frac{1}{x_a(1-x_a)},
\]
which is similar to the speed measure in a two-allele model with no mutation.

Theorem \ref{thm:1} also provides a proof of the approximate density (\ref{density:20}).
\begin{theorem}
The stationary density for pairs of frequencies $X_a$ and $X_b=1-X_a$, in the Wright-Fisher diffusion with generator (\ref{gen:0}) when $\theta \to 0$ is (\ref{density:20}) to $\mathcal{O}(\theta)$.
\end{theorem}
\begin{proof}
Denote 	the density  Eq.~(\ref{density:20}) by $f_{ab}(x_a)$.
Although $f_{ab}(x_a)$ is not integrable over $(0,1)$, $x_a(1-x_a)f_{ab}(x_a)$ is integrable and sampling distributions
uniquely determine $x_a(1-x_a)f_{ab}(x_a)$. Therefore it is sufficient to note the easy calculation that, to $\mathcal{O}(\theta)$, for $n_a,n_b\geq 1$
\begin{equation*}
\int_0^1{n\choose n_a}x^{n_a}(1-x)^{n_b}f_{ab}(x_a) dx_a
= \theta\Big (\pi_aP_{ab}\frac{1}{n_b} + \pi_bP_{ba}\frac{1}{n_a}\Big ).
\end{equation*}
\qed
\end{proof}

Theorem \ref{thm:1} is derived for fixed $n$, so the terms in (\ref{thm:a}) and (\ref{thm:b}) 
of $\mathcal{O}(\theta^2)$ are not necessarily small as $n \to \infty$. The pgf of the number of mutations in a sample of $n$ is 
\[\prod_{j=1}^{n-1}\Big (1 - \frac{\theta}{j}(s-1)\Big )^{-1}\]
 which is asymptotic to a Poisson pgf $e^{\theta \log (n) (s-1)}$.
 The probability of greater than one mutation in a coalescent tree is therefore $\mathcal{O}(\theta^2\log (n)^2)$ as $\theta \to 0$ and $n \to \infty$. The formulae (\ref{thm:a}) and (\ref{thm:b}) really then hold provided $\theta\log (n)$ is small. 

\citet{JS2011} in Lemma 1 derive an exact formula, when $\theta$ is not necessarily small, for the probability that there is one mutation from $a$ to $b$ in a coalescent tree resulting in a sample configuration $\bm{n}=n_a\bm{e}_a + n_b\bm{e}_b$.  This probability is
\begin{equation}
\theta P_{ab}\frac{(n-1)!}{~~(1+\theta)_{(n-1)}}
\sum_{l=1}^{n_a}
\frac{
{n_a-1\choose l-1}
}
{
{n-1 \choose l}
}
\frac{1}{l+\theta}.
\label{Song:0}	
\end{equation}
An expansion of (\ref{Song:0}) to the first order in $\theta$ is
\[
\theta P_{ab}\frac{1}{n_b} + \mathcal{O}(\theta^2),
\]
which agrees with (\ref{thm:b}).
\citet{BKS2012} derive formulae for the leading coefficients in $p(\bm{n})$ of powers of $\theta$ depending on the number of mutations in a tree.  If there are $j$ mutations then the leading term is $\mathcal{O}(\theta^j)$. Eq. (\ref{thm:b}) is contained in their Theorem 1.
\begin{remark}
A Moran model in continuous time has a fixed population size of $N$ genes. The behaviour of the population frequencies in the stationary distribution of this model turns out to be the same as the stationary distribution of $N$ genes in a Wright-Fisher diffusion process. Thus Theorem \ref{thm:1} holds for the population frequencies taking $n=N$. Birth-death events in the model occur at rate $\lambda$ when an individual is chosen at random to reproduce a child (and continue as a parent) and an individual is chosen to die from the individuals before reproduction. (This could be the parent.) Mutations occur according to a Poisson process at rate $\theta/2$ along the edges of ancestral trees, independently from reproduction. The distribution of the ancestral tree of $n$ genes is the same as a coalescent tree if the time scale is $\lambda = N(N-1)/2$.  The limit relative frequencies in a discrete Wright-Fisher model and a Moran model with this rate form a diffusion process with generator (\ref{gen:0}).
\end{remark}
\begin{example} A stepwise mutation model has $d$ allele types, with $P_{i,i+1} = \alpha, P_{i,i-1} = 1 - \alpha$, $P_{ij} = 0$ if $|i-j| >1$, where indices are read around a circle. The stationary distribution is $\pi_i=1/d$, $i=1,\ldots, d$.  Then for small $\theta$ the probability a sample contains adjacent types $n_i,n_{i+1}$ is
\[
\frac{\theta}{d}\Bigg ( \frac{\alpha}{n_{i+1}} + \frac{1-\alpha}{n_i}\Bigg ) + \mathcal{O}(\theta^2)
\]
and the probability of a sample configuration that has non-adjacent types is 
$\mathcal{O}(\theta^2)$. $d$ is usually thought of as large in this model.
\end{example}
\begin{example}
A model of $L$ completely linked sites has $K$ types at each site with $d=K^L$ and mutation matrix
\[
P  = \frac{1}{L}\Big ( M\otimes I \cdots \otimes I + I\otimes M \cdots \otimes I +
\cdots +I \otimes \cdots \otimes M\Big ),
\]
where $M$ is a $K\times K$ transition probability matrix for mutations at a site, with stationary distribution $\bm{\pi}^{(M)}$ and $\otimes$ denotes direct product. $\theta/2$ is the mutation rate per sequence, or $\theta/2L$ per site. The stationary distribution of $P$ is $\bm{\pi}^{(M)}\otimes \cdots \otimes \bm{\pi}^{(M)}$, which gives the probability of a fixed configuration in a sample of $n$ to constant order. The only configurations which have a probability of $\mathcal{O}(\theta)$ are those with two types which differ at just one site, say site $r$ with types $a,b$. Then
the sequences are
\begin{eqnarray*}
i &=& i_1,\ldots,a,i_{r+1},\ldots,i_L\phantom{.}
\nonumber \\
j &=& i_1,\ldots,b,i_{r+1},\ldots,i_L.
\end{eqnarray*}
The probability of a configuration $i,j$ is then
\[
\frac{\theta}{2}\prod_{k\ne r}\pi_{i_k}^{(M)}
\Bigg ( \frac{\pi_a^{(M)}M_{ab}}{n_j} +  \frac{\pi_b^{(M)}M_{ba}}{n_i} \Bigg ) + \mathcal{O}(\theta^2)
\]

\end{example}

Inspection of the proof of Theorem \ref{thm:1} shows that a similar theorem holds for general exchangeable binary trees. In such trees if there are $j$ edges, the probability of coalescence of a particular pair is ${j\choose 2}^{-1}$. 
The combinatorial nature of these trees is the same as in a coalescent tree, but the edge lengths $T_2,\ldots, T_n$ have a general distribution. The edge lengths do not need to be independent. 
Examples of general binary trees are in \citet{GT1988, GT2003}.

 A coalescent tree with a non-homogeneous population size is an example of a general binary tree. Suppose the population size is  $N(t)= N_0/\nu (t)$ at time $t$ back.
Denote $S^\nu_j= T^\nu_j + \cdots + T^\nu_n$, with $S^\nu_{n+1}=0$.
$\{S^\nu_j\}$ form a reverse Markov Process with transition density
of $S^\nu_j$ given $S^\nu_{j+1}=t$ of
\[f(s;t)
= {j\choose 2}\nu(s)
\exp\Bigl (-{j\choose 2}\int_t^s \nu (u)du\Bigr ), s>t.
\]
$\mathbb{E}\big [T_j\big]$ does not have a simple form in these trees.

The tree of a pure birth process is another example which has detail in this paper.

\begin{theorem}\label{thm:2}
The probability of a sample configuration in a general exchangeable binary tree satisfies the following.
 As $\theta \to 0$, for $a\ne b \in \{1,\ldots,d\}$,
\begin{eqnarray}
p(n\bm{e}_a;\theta) &=& \pi_a\Big (1 - \frac{\theta}{2}(1-P_{aa})\sum_{l=2}^{n}l\mathbb{E}\big [T_l\big ]\Big ) + \mathcal{O}(\theta^2)
\label{thm:ga} \\
&=&\pi_a\Big (1 + \gamma_{aa}\sum_{l=2}^{n}l\mathbb{E}\big [T_l\big ]\Big ) + \mathcal{O}(\theta^2)
\end{eqnarray}
\begin{eqnarray}
p(n_a\bm{e}_a+n_b\bm{e}_b;\theta) &= &
\frac{\theta}{2}\frac{(n_a-1)!(n_b-1)!}{(n-1)!}
\nonumber \\
&&~~\times
\Bigg (\pi_aP_{ab}\sum_{k=2}^{n_a+1}{n-k\choose n_b-1}k(k-1)\mathbb{E}\big [T_k\big ]
\nonumber \\
&&~~ +
\pi_bP_{ba}\sum_{k=2}^{n_b+1}{n-k\choose n_a-1}k(k-1)\mathbb{E}\big [T_k\big ]\Bigg ) + \mathcal{O}(\theta^2)
\nonumber \\
\label{thm:gb} \\
p(\bm{n};\theta) &=& \mathcal{O}(\theta^2)\text{~if~}\bm{n}\text{~has~}>2\text{~non-zero~entries}
\label{thm:gc}
\end{eqnarray}
\end{theorem}
\begin{example} 
\noindent
\emph{Pure birth process tree}

In a pure birth process $\{X_t\}_{t\geq 0}$ of counts of individuals in continuous time individuals split at rate $\lambda_x$ when $X_t=x$.  In this example $x_0=1$ and we take the type of the initial individual to be $a$ with probability $\pi_a$. The type of individuals is defined by mutations occurring in the tree at rate $\theta/2$ and transitions according to $P$. The tree is an exchangeable binary tree with $T_2,\ldots, T_N$ independent exponential random variables with rates $\lambda_2,\ldots,\lambda_n$.
\begin{corollary}
In a pure birth process where individuals reproduce independently at rate $\lambda$, $\lambda_n=n\lambda$ and
\begin{eqnarray}
p(n\bm{e}_a;\theta) &=& \pi_a\Big (1 - (n-1)\frac{\theta}{2\lambda}(1-P_{aa})\Big ) + \mathcal{O}(\theta^2)
\nonumber \\
p(n_a\bm{e}_a+n_b\bm{e}_b;\theta) &= &
\frac{\theta}{2\lambda}\Bigg (\pi_aP_{ab}\frac{n}{n_b(n_b+1)} + \pi_bP_{ba}\frac{n}{n_a(n_a+1)}\Bigg ) + \mathcal{O}(\theta^2).
\nonumber \\
\label{pure:0}
\end{eqnarray}
\end{corollary}
\begin{proof}
Note that 
$
\mathbb{E}\big [T_j\big ] = (\lambda j)^{-1}.
$
Calculation of $p(n\bm{e}_a;\theta)$ is elementary. 
Calculation of $p(n_a\bm{e}_a+n_b\bm{e}_b;\theta)$ depends on the identity
\begin{eqnarray*}
\sum_{k=2}^{n_a+1}{n-k\choose n_b-1}(k-1)
&=& \sum_{k=2}^{n_a+1}\Bigg (n{n-k\choose n_b-1} - (n-k+1){n-k\choose n_b-1}\Bigg )
\nonumber \\
&=& \sum_{k=2}^{n_a+1}\Bigg (n{n-k\choose n_b-1} - n_b{n-k+1\choose n_b}\Bigg )
\nonumber \\
&=& n{n-1\choose n_b} - n_b{n\choose n_b+1}
\nonumber \\
&=& \frac{(n-1)!}{(n_a-1)!(n_b-1)!}\cdot \frac{n}{n_b(n_b+1)}.
\end{eqnarray*}
and a similar identity with $a$ and $b$ interchanged. The principle is similar to that in the proof of Theorem \ref{thm:1}.
\end{proof}
\end{example}
\section{Discussion}
An approximation to the sampling distribution in the stationary distribution of a Wright-Fisher diffusion model with general mutation rates has been found in Theorem \ref{thm:1}. The coalescent process is a dual process to the diffusion and this is exploited to show that the approximation is equivalent to considering at most one mutation in the coalescent tree. \citet{BT2016,BT2017} have previously derived approximate expressions for the stationary distribution and the sampling distribution using a probability flux argument. The coalescent argument in this paper provides a neat proof and a probabilistic insight into the approximation.
An approximation to the stationary distribution in the diffusion is characterized by the sampling approximations.
The coalescent argument is extended to general random exchangeable binary trees in Theorem \ref{thm:2}, for example coalescent trees in a variable population size model, or pure birth process trees. 
\begin{acknowledgements}
This research was done when Robert Griffiths visited the Mathematical Sciences Instutite, Australian National University in November and December 2017. He thanks the Instutite for their support and hospitality.
\end{acknowledgements}


\begin{thebibliography}{99}

  \bibitem[Bhaskar, Kamm and Song(2012)]{BKS2012}
   Bhaskar, A., Kamm, J. A., Song, Y. S. (2012). Approximate sampling formulae for general finite-alleles models of mutation.  Adv.  Appl. Probab. 44,  408--428.


\bibitem[Burden and Tang(2016)]{BT2016} Burden, C. J. and Tang, Y. (2016).  An approximate stationary solution for multi-allele diffusion with low mutation rates. Theor. Popul. Biol. 112, 22--32.

\bibitem[Burden and Tang(2017)]{BT2017} Burden, C. J. and Tang, Y. (2017). Rate matrix estimation from site frequency data.
Theor. Popul. Biol. 113, 23--33.
%
\bibitem[Burden and Griffiths(2018)]{BG2018} Burden, C. J. and Griffiths, R. C. (2018). 
Stationary distribution of a 2-island 2-allele Wright-Fisher diffusion model with slow mutation and migration rates.  arXiv:1802.07415

\bibitem[De Maio, Schrempf, Kosiol(2015)]{DSK2015}
De Maio, N, Schrempf, D., Kosiol,C. (2015). PoMo: An allele frequency based approach for species tree estimation. Syst. Biol. 64, 1018--1031.

\bibitem[De Iorio and Griffiths(2004)]{DG2004}
De Iorio, M. and Griffiths, R. C. (2004). Importance sampling on coalescent histories. I. Adv. Appl. Prob. 36, 417--433. 

\bibitem[Etheridge(2011)]{E2011}
Etheridge, A. (2011). Some Mathematical Models from Population Genetics: {\'E}cole D\'{\'E}t{\'e} de Probabilités de Saint-Flour XXXIX-2009, Springer-Verlag, Berlin, Heidelberg.


\bibitem[Griffiths and Tavar{\'e}(1988)]{GT1988}
Griffiths, R.C. and Tavar{\'e}, S. (1998). The age of a mutation in a general coalescent tree. Stochastic Models. 14 273--295. 

\bibitem[Griffiths and Tavar{\'e}(2003)]{GT2003}
Griffiths, R.C. and Tavar{\'e}, S. (2003). The genealogy of a neutral mutation. In: Green, P.J., Hjort, N.L. and Richardson, S. (Eds.), Highly Structured Stochastic Systems. Oxford University Press, Oxford, pp. 393--413. 

 \bibitem[Jenkins and Song(2010)]{JS2010}
   Jenkins, P. A., Song, Y. S. (2010). An asymptotic sampling formula for the coalescent with recombination. Ann. Appl. Probab., 20, 1005--1028.
   
    \bibitem[Jenkins and Song(2011)]{JS2011}
   Jenkins, P. A., Song,  Y. S. (2011). The effect of recurrent mutation on the frequency spectrum of a segregating site and the age of an allele, Theor. Popul. Biol., 80, 158--173.


\bibitem[Kimura(1969)]{K1969} 
Kimura, M. (1969) The number of heterozygous nucleotide sites maintained in a finite population due to steady flux of mutations. Genetics 61 893--903. 

\bibitem[Kingman(1982)]{K1982} Kingman, J.F.C. (1982). The coalescent.
Stochastic Process. Appl. 13, 235--248.

\bibitem[RoyChoudhury and Wakeley(2010)]{RW2010} RoyChoudhury, A. and Wakeley, J. (2010). Sufficiency of the number of segregating sites in the limit under finite-sites mutation.  
Theor. Popul. Biol. 78, 118--122.


\bibitem[Schrempf and Hobolth(2017)]{SH2017}
Schrempf, D. and Hobolth, A. (2017). An alternative derivation of the stationary distribution of the multivariate neutral Wright-Fisher model for low mutation rates with a view to mutation rate estimation from site frequency data. 
Theor. Popul. Biol. 114, 88--94.

\bibitem[Vogl and Clemente(2012)]{VC2012} Vogl, C., Clemente, F. (2012). The allele-frequency spectrum in a decoupled Moran model with mutation, drift, and 
directional selection, assuming small mutation rates
{Theor. Popul. Biol.} 81, 197--209.

\bibitem[Vogl(2014)]{V2014} Vogl, C. (2014). Estimating the scaled mutation rate and mutation bias with site frequency data.
{Theor. Popul. Biol.} 98, 19--27.

\bibitem[Vogl and Bergman(2015)]{VB2015} Vogl, C., Bergman J. (2015). Inference of directional selection and mutation parameters assuming equilibrium.
{Theor. Popul. Biol.} 106, 71--82.

\bibitem[Zeng(2010)]{Z2010}
Zeng, K. (2010) A simple multiallele model and its application to identifying preferred-unpreferred codons using polymorphism data.  {Mol. Biol. Evol.} 27 1327--1337.
%

\end{thebibliography}
\end{document}